\documentclass[10pt,conference]{IEEEtran}\IEEEoverridecommandlockouts
\usepackage{epsfig,graphics,subfigure,psfrag,amsmath,cases}
\usepackage{latexsym,amssymb,amsmath,epsfig,subfigure,algorithm,amsthm}
\usepackage{algorithmic}
\usepackage{color}
\usepackage{ctable}
\usepackage{url}
\usepackage{scrtime}
\usepackage{bm}
\usepackage{graphicx}
\usepackage{epstopdf}
\usepackage{units}
\usepackage[font=footnotesize]{caption}
\usepackage{cite} 

\author{ \IEEEauthorblockN{Rania Morsi\IEEEauthorrefmark{1}, Elena Boshkovska\IEEEauthorrefmark{1}, Esmat Ramadan\IEEEauthorrefmark{1}, Derrick Wing Kwan Ng\IEEEauthorrefmark{2}, and Robert Schober\IEEEauthorrefmark{1} 
}
\IEEEauthorblockA{\IEEEauthorrefmark{1}Friedrich-Alexander-University Erlangen-N\"urnberg (FAU), Germany\\} 
 \IEEEauthorblockA{\IEEEauthorrefmark{2}The University of New South Wales, Australia
    \\}
 \\[-4ex]
}

\title{On the Performance of Wireless Powered Communication With Non-linear Energy Harvesting}

\date{\thistime,\,\today}

\theoremstyle{plain}
\newtheorem{corollary}{Corollary}

\newtheorem{proposition}{Proposition}

\theoremstyle{definition}

\theoremstyle{remark}
\newtheorem{remark}{Remark}

\newcommand{\norm}[1]{\lVert#1\rVert}

\newcommand{\V}[1]{\boldsymbol{#1}} 

\newcommand{\e}{{\rm e}}
\newcommand{\dd}{{\rm d}}
\newcommand{\PR}{P_{\rm R}}
\newcommand{\PEH}{P_{\rm{EH}}}
\newcommand{\EEH}{E_{\rm{EH}}}
\newcommand{\PWD}{P_{\rm{WD}}}
\newcommand{\pout}{P_{\rm{out}}}
\newcommand{\Pt}{P_{\rm{t}}}
\newcommand{\pr}{\mathbb{P}}
\newcommand{\gammathr}{\gamma_{\rm thr}}
\newcommand{\quotes}[1]{``#1''}

\begin{document}

\maketitle
\begin{abstract}
In this paper, we analyze the performance of a time-slotted multi-antenna wireless powered communication (WPC) system, where a wireless device first harvests radio frequency (RF) energy from a power station (PS) in the downlink to facilitate information transfer to an information receiving station (IRS) in the uplink.  The main goal of this paper is to provide insights and guidelines for the design of practical WPC systems. To this end, we adopt a recently proposed parametric non-linear RF energy harvesting (EH) model, which has been shown to accurately model the end-to-end non-linearity of practical RF EH circuits. In order to enhance the RF power transfer efficiency,  maximum ratio transmission is adopted at the PS to focus the energy signals on the wireless device. Furthermore, at the IRS, maximum ratio combining is used. We analyze the outage probability and the average throughput of information transfer, assuming Nakagami-$m$ fading uplink and downlink channels. Moreover, we study the system performance as a function of the number of PS transmit antennas, the number of IRS receive antennas, the transmit power of the PS, the fading severity, the transmission rate of the wireless device, and the EH time duration. In addition, we obtain a fixed point equation for the optimal transmission rate and the optimal EH time duration that maximize the asymptotic throughput for high PS transmit powers. All analytical results are corroborated by simulations. 
\end{abstract}

\renewcommand{\baselinestretch}{1}
\large\normalsize
\vspace{-0.2cm}
\section{Introduction}
\label{s:introduction}
The limited lifetime of wireless communication devices has motivated the scavenging of energy from renewable energy sources to ensure a perpetual energy supply and sustainable network operation. However, opportunistic energy harvesting (EH) from conventional renewable energy sources such as solar and wind energy  is  uncontrollable, weather dependent, and not available indoors. On the other hand, radio frequency (RF)-based wireless power transfer (WPT) exploits the far-field properties of electromagnetic waves and facilitates stable wireless charging that can be provided on-demand \cite{RF_EH_networks_survey_2015,JR:Energy_harvesting_circuit,self_calibrating_EH_circuit2013,Krikidis2014,WPC_survey_Zhang_2016,
WPC_beamforming_Zhang,EE_WPC_Schober_2016,WIET_beamforming_feedback_2014,CR:Perf_Analys_WPT_Nakagami_linear,
perfom_analy_WPC_beamforming_linear_Vucetic2016,Morsi_storage_arxiv,
Letter_non_linear,Journal_non_linear,L:Perf_Analys_Nakagami_nonlin,J:Perf_Analys_WPR_MIMO_non_lin,waveform_nonlinear_WPT_Clerckx_Zhang2017}.

However, due to the high propagation loss, WPT is only applicable for charging low-power devices over a short distance.
One way to improve the energy transfer efficiency, and therefore to extend the WPT range, is by exploiting multiple transmit antennas at the RF power source to focus the energy signals at the EH receivers via beamforming
\cite{WPC_beamforming_Zhang,WIET_beamforming_feedback_2014,perfom_analy_WPC_beamforming_linear_Vucetic2016, Letter_non_linear,Journal_non_linear,L:Perf_Analys_Nakagami_nonlin,waveform_nonlinear_WPT_Clerckx_Zhang2017}.
Moreover, a considerable amount of work has been devoted to the optimization of the efficiency of the RF EH circuits, which convert the collected RF energy to electrical direct current (DC) energy at the EH receivers, see e.g. \cite{RF_EH_networks_survey_2015,JR:Energy_harvesting_circuit,self_calibrating_EH_circuit2013}.

In this paper, we consider wireless powered communication (WPC), where wireless devices harvest RF energy and exploit the harvested energy to transfer information to their designated receivers \cite{WPC_survey_Zhang_2016}.
The design, analysis, and optimization of different WPC systems have recently received considerable attention 
\cite{WPC_beamforming_Zhang,EE_WPC_Schober_2016,WIET_beamforming_feedback_2014,CR:Perf_Analys_WPT_Nakagami_linear,
perfom_analy_WPC_beamforming_linear_Vucetic2016,Morsi_storage_arxiv,
Letter_non_linear,Journal_non_linear,L:Perf_Analys_Nakagami_nonlin,J:Perf_Analys_WPR_MIMO_non_lin,waveform_nonlinear_WPT_Clerckx_Zhang2017}. 
For example, \cite{WPC_beamforming_Zhang} and \cite{EE_WPC_Schober_2016} considered a multi-user WPC system that employed a time-division based harvest-then-transmit protocol and jointly optimized the users' time and power allocation to maximize the minimum user throughput and the system energy efficiency, respectively. In addition, several works focused on the performance analysis of WPC systems \cite{CR:Perf_Analys_WPT_Nakagami_linear,perfom_analy_WPC_beamforming_linear_Vucetic2016,Morsi_storage_arxiv}. For example, \cite{CR:Perf_Analys_WPT_Nakagami_linear} analyzed the outage probability and the average error rate of a WPC system in Nakagami-$m$ fading. In \cite{perfom_analy_WPC_beamforming_linear_Vucetic2016}, the average throughput of a WPC system was analyzed for delay-limited and delay-tolerant transmission.  Moreover, \cite{Morsi_storage_arxiv} studied the impact of energy storage on the outage probability of WPC systems.

Most of the literature on WPC systems, e.g. \cite{WPC_beamforming_Zhang,EE_WPC_Schober_2016,WIET_beamforming_feedback_2014,CR:Perf_Analys_WPT_Nakagami_linear
,perfom_analy_WPC_beamforming_linear_Vucetic2016,Morsi_storage_arxiv}, adopts a linear RF EH model. This model assumes that the harvested DC power increases linearly and without bound with the RF power arriving at the EH circuit. It also assumes zero sensitivity, which is the minimum amount of RF input power necessary for the EH circuit to operate. 
However, experiments with practical RF EH circuits have shown that their input-output characteristic is highly non-linear. This is because rectifying EH circuits employ non-linear circuit elements, such as diodes and transistors \cite{RF_EH_networks_survey_2015,JR:Energy_harvesting_circuit,self_calibrating_EH_circuit2013}.  In particular, a practical RF EH circuit is typically characterized by a non-zero sensitivity for low input powers and saturation of the harvested power for high input powers. Hence, in order to accurately design WPT systems, a non-linear EH model is necessary. To this end,  the authors of \cite{Letter_non_linear} proposed a parametric non-linear EH model that accurately matches  measurement data from several practical RF EH circuits  \cite{RF_EH_networks_survey_2015,JR:Energy_harvesting_circuit,self_calibrating_EH_circuit2013}. Later, the authors of \cite{L:Perf_Analys_Nakagami_nonlin} and \cite{J:Perf_Analys_WPR_MIMO_non_lin} adopted a simpler piece-wise linear transfer function which also accounts for the saturation effect of practical EH circuits, but it cannot fully model the joint effect of circuit sensitivity and the current leakage \cite{RF_EH_networks_survey_2015,JR:Energy_harvesting_circuit,self_calibrating_EH_circuit2013}. 
 Furthermore, another interesting line of research in \cite{waveform_nonlinear_WPT_Clerckx_Zhang2017} studies the design of transmit signal waveforms that maximize the overall power transfer efficiency of WPT systems, considering a non-linear EH model. It is shown in \cite{Journal_non_linear} and \cite{waveform_nonlinear_WPT_Clerckx_Zhang2017} that,  due to the mismatch between the linear EH model and the non-linear behaviour of practical EH circuits, optimizing WPC systems based on the linear EH model leads to a significant performance degradation compared to  designs based on the non-linear EH models from \cite{Letter_non_linear} and \cite{waveform_nonlinear_WPT_Clerckx_Zhang2017}.
However, the design in  \cite{waveform_nonlinear_WPT_Clerckx_Zhang2017} is based on an analytical model for a specific rectifier, which does not capture impedance and input power mismatches nor the use of a general rectifier (e.g. with multiple diodes or transistor-based rectifiers). Hence, in this paper, we adopt the non-linear EH model from \cite{Letter_non_linear}, since it is a parameter-based model whose parameters can be easily obtained by accurate curve-fitting to the input-output characteristic of practical EH circuits.

In this paper, we perform a comprehensive analysis of a multi-antenna WPC systems based on the non-linear EH model in \cite{Letter_non_linear}. We consider a time-division based harvest-then-transmit protocol, where in each time slot, a wireless device (WD) first harvests RF energy from a power station (PS) and then uses the harvested energy to transmit data to an information receiving station (IRS). We analyze the outage probability and the average throughput of the wireless information transfer (WIT) link in Nakagami-$m$ fading. Furthermore, we evaluate the system performance as a function of the number of antennas, the PS transmit power, the fading severity, the transmission rate of the WD, and the EH time duration.\\
\emph{Notations:}
$\mathbb{C}^{N\times 1}$ represents the set of all column vectors of size $N$ with complex-valued entries. $|x|^2$ denotes the magnitude squared of complex number $x$ and $(\cdot)^*$ denotes the complex conjugate operator. $\norm{\V{x}}$ denotes the Euclidean norm of vector $\V{x}$ and $x_n$ denotes the $n$th element of vector $\V{x}$.
$\frac{\partial^{n}f }{\partial s^n}$ denotes the $n{\rm th}$ order partial derivative of function $f$ with respect to variable $s$. $\sim$ stands for \quotes{is distributed as}. ${\rm Gamma}(m,\lambda)$ is the Gamma distribution with shape parameter $m$ and rate parameter $\lambda$. $\Gamma(\cdot)$ is the Gamma function defined as $\Gamma(x)\!=\!(x-1)!$ for a positive integer $x$ and as $\Gamma(x)\!=\!\int_{0}^{\infty}t^{x-1}\e^{-t} \dd t$ for a positive real number $x$. $\Gamma(m,x)\!=\!\int_x^\infty t^{m-1}\e^{-t}\dd t$ and $\gamma(m,x)\!=\!\int_0^x t^{m-1}\e^{-t}\dd t$  are the upper and the lower incomplete Gamma functions, respectively. $W_{\alpha,\beta}(\cdot)$ is the Whittaker W function defined in \cite[Eq. 9.222.1]{table_of_integrals_Ryzhik}. $U(\alpha,\gamma,z)=\frac{1}{\Gamma(\alpha)}\int_0^\infty \e^{-z t}t^{\alpha-1}(1+t)^{\gamma-\alpha-1}\dd t$ is the confluent hypergeometric function defined in \cite[Eq. 9.211.4]{table_of_integrals_Ryzhik}. Finally, $\pr$ denotes the probability of an event.\vspace{-0.1cm}
\section{System Model and Preliminaries}
\label{s:system_model}
\subsection{System Model}
\begin{figure}[!t] 
\centering
\includegraphics[width=0.4\textwidth]{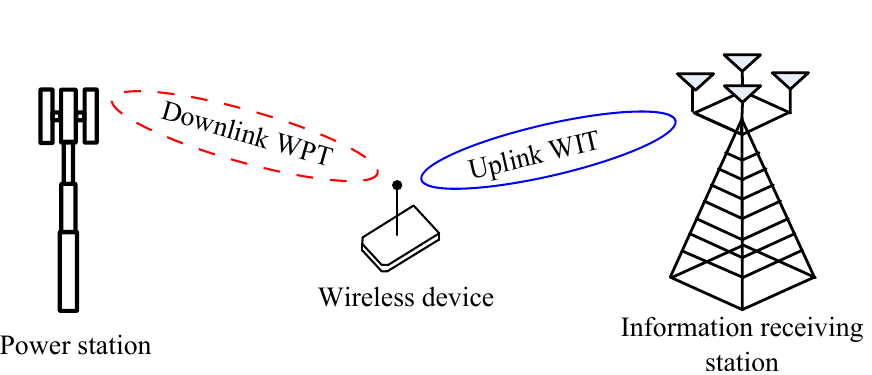}
\caption{A WPC system with DL WPT and UL WIT.\\[-2ex]}
\label{fig:system_model}
\end{figure}
We consider the time-slotted WPC system shown in Fig. \ref{fig:system_model}, where a PS transfers RF energy to a WD in the downlink (DL) to facilitate WIT to an IRS in the uplink (UL). It is assumed that the PS and the IRS are equipped with a fixed power source and with $N_1\!\geq\!1$ and $N_2\!\geq\!1$ antennas, respectively. However, the WD is assumed to be a simple low-power node that is solely powered by the RF signals broadcasted by the PS and is equipped with a single antenna due to energy and space constraints. 
Furthermore, we assume that the WD operates in a half-duplex in-band mode, where DL WPT and UL WIT occupy two orthogonal subslots of a time slot but utilize the same frequency band \cite{WPC_survey_Zhang_2016}. In particular, we adopt the harvest-then-transmit protocol, where in each time slot of duration $T$, the WD first harvests RF power from the PS for a time duration of $\tau T$ and then uses the harvested energy to transmit its backlogged data to the IRS for a time duration of $(1-\tau)T$, where $\tau$ is the EH time factor that satisfies $0<\tau<1$. More specifically, the WD is assumed to have limited storage capacity and therefore it consumes the harvested energy fully on a slot-by-slot basis. Moreover, the WD is assumed to have no knowledge of the UL channel state information (CSI) nor of the amount of harvested energy. Hence, power control at the WD is not possible. 
\begin{figure}[!t] 
\centering
\includegraphics[width=0.36\textwidth,height=0.16\textheight]{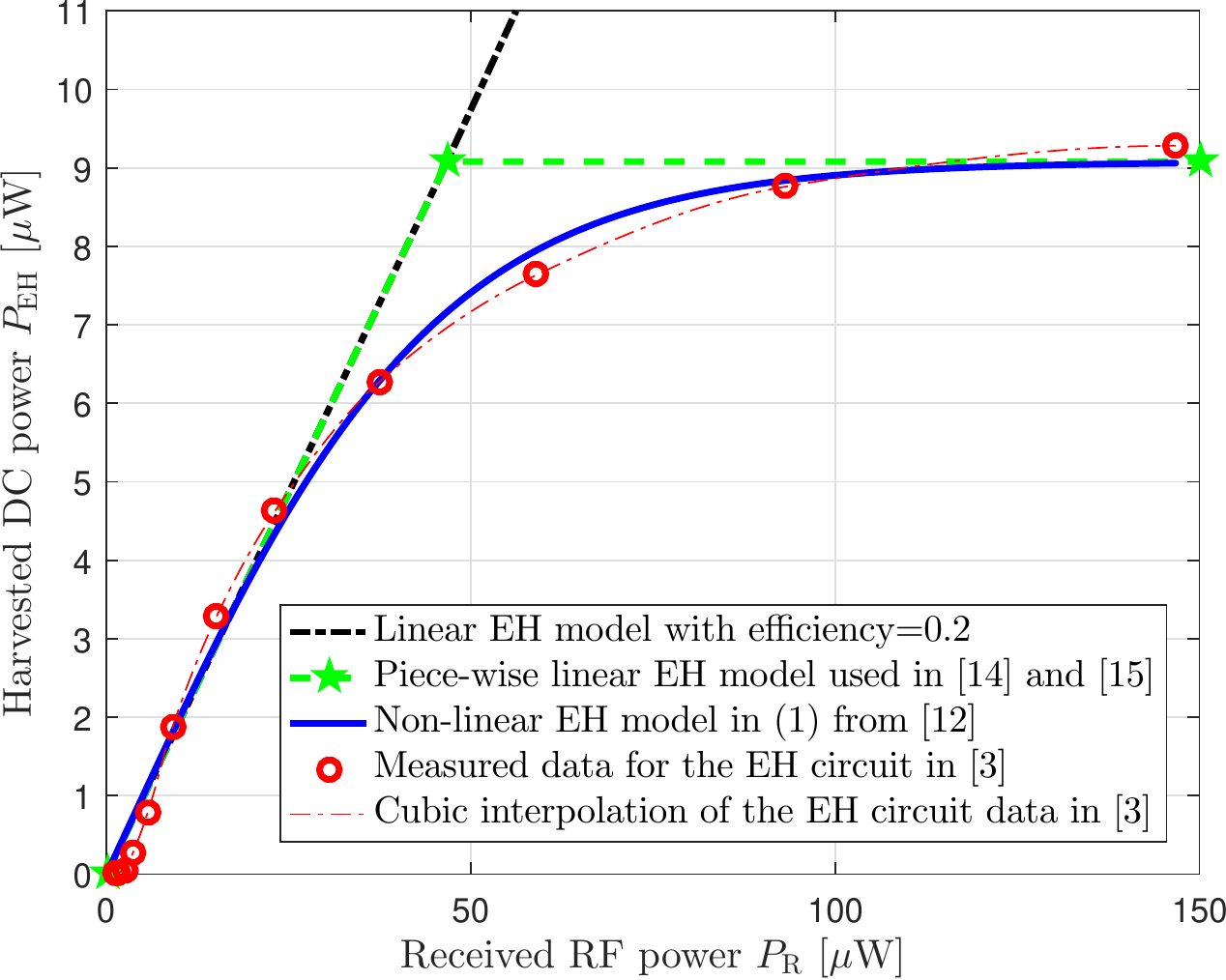}
\caption{A comparison between the RF-to-DC power transfer functions of different EH models. Parameters $a\!=\!47.083\times 10^{3}$, $b\!=\!\unit[2.9]{\mu W}$, and $M\!=\!\unit[9.079]{\mu W}$ of the non-linear EH model in (\ref{eq:EH_non_linear}) were obtained using a standard curve fitting tool. The data of the EH circuit is obtained from \cite[Fig. 5(b) for a load resistance of $1\,$M$\Omega$]{self_calibrating_EH_circuit2013}.\vspace{-0.2cm}}
\label{fig:transfer_functions}
\end{figure}
\subsection{Energy Harvesting Model}
Unlike the conventional linear EH model commonly used in the literature, in this paper, we adopt the non-linear parametric EH model from \cite{Letter_non_linear} to characterize the end-to-end non-linear behaviour of RF EH circuits. In this model, the total harvested power at the WD is given by \cite{Letter_non_linear}
\begin{equation}
 \PEH= \frac{\frac{M}{1+\exp(-a(\PR-b))}  -M\Omega}{\!\!1 - \Omega}=M\frac{1-\e^{-a\PR}}{\,1+\e^{-a(\PR-b)}},
 \label{eq:EH_non_linear}
 \end{equation}
where $\Omega = 1/(1+\e^{ab})$ and $\PR$ is the received RF power at the WD. Parameters $M$, $a$, and $b$ in (\ref{eq:EH_non_linear}) capture the joint effects of various non-linear phenomena caused by hardware limitations. More specifically, $M$ represents the maximum power that can be harvested by the RF EH circuit, $a$ and $b$ are related to different physical hardware phenomena, such as the circuit sensitivity and current leakage. This model has been shown to accurately match  measurement data from various practical RF EH circuits \cite{RF_EH_networks_survey_2015,JR:Energy_harvesting_circuit,self_calibrating_EH_circuit2013}. For example, in Fig. \ref{fig:transfer_functions}, we compare the power transfer function of the non-linear EH model from \cite{Letter_non_linear}, the conventional linear EH model, and the piece-wise linear EH model used in \cite{L:Perf_Analys_Nakagami_nonlin} and \cite{J:Perf_Analys_WPR_MIMO_non_lin}. We also show the measurement data for an RF EH circuit read from  \cite[Fig.5(b) for a load resistance of $1\,$M$\Omega$]{self_calibrating_EH_circuit2013} and its cubic interpolation. The interpolated transfer function will be used in Section \ref{s:simulation_results} for the evaluation of the system performance obtained with the EH circuit in \cite{self_calibrating_EH_circuit2013} via simulation.
\subsection{Channel Model}
The DL WPT and UL WIT channel coefficient vectors are denoted by $\V{h}\in\mathbb{C}^{N_1\times 1}$ and $\V{g}\in\mathbb{C}^{N_2\times 1}$, respectively. We assume that $\V{h}$ and $\V{g}$ capture the joint effect of the large scale path loss and the small-scale multipath fading. Furthermore, all channels are assumed to be quasi-static flat block fading, i.e., the channels remain constant over one time slot, but may vary from one slot to the next. Moreover, the $N_1$ WPT links are assumed to be independent and identically distributed (i.i.d.) Nakagami-$m$ fading. Hence, the channel power gain of each link is Gamma distributed with mean $\mu_1$, shape parameter $m_1$, and rate parameter $\lambda_1=m_1/\mu_1$, i.e., $|h_n|^2 \sim {\rm Gamma}(m_1,\lambda_1)$, $\forall n\in1,\ldots,N_1$. Similarly, the $N_2$ WIT links are assumed to be i.i.d. Nakagami-$m$ fading with parameters $\mu_2$, $m_2$, and $\lambda_2=m_2/\mu_2$, i.e., $|g_n|^2 \sim {\rm Gamma}(m_2,\lambda_2)$,  $\forall n\in1,\ldots,N_2$. We assume integer shape parameters, i.e., $m_1,m_2\in\{1,2,\ldots\}$. The Nakagami-$m$ fading model is adopted since field measurements have shown that it provides a good fit to outdoor and indoor multipath propagation. It is also a general fading model that reduces to Rayleigh fading for $m=1$ and can approximate Ricean fading \cite{Statistical_channel_model}, \cite[Eq. (2.26)]{Digital_comm_fading_Alouini2005}. 

In order to enhance the energy efficiency of the DL WPT, the PS utilizes its multiple antennas to focus the energy signal at the WD via beamforming. In particular, the PS performs maximum ratio transmission (MRT) with a beamforming vector $\V{w}_1\!=\!\V{h}^*/\norm{\V{h}}$, since this is known to maximize the amount of harvested energy at the WD \cite{WIET_beamforming_feedback_2014}. Moreover, the IRS performs maximum ratio combining (MRC) with a combining weight vector $\V{w}_2\!=\!\V{g}^*$ to maximize the instantaneous signal-to-noise ratio (SNR) of the combined signal at the IRS. We note that MRT and MRC require knowledge of the DL and UL channels at the PS and IRS, respectively. We assume perfect CSI knowledge to obtain an upper bound on the system performance. 
\subsection{Equivalent Single Antenna System}
The energy beamforming with MRT in the DL WPT phase leads to an effective DL channel power gain $v_1$ equal to the sum of the channel power gains of the $N_1$ individual DL channels, i.e., $v_1=\norm{\V{h}}^2$ \cite{Digital_comm_fading_Alouini2005}. Similarly, the equivalent UL channel power gain after MRC is $v_2=\norm{\V{g}}^2$.
Hence, assuming Nakagami-$m$ fading for the individual links, the equivalent channel power gains follow a Gamma distribution, i.e., $v_i\sim {\rm Gamma}(m_iN_i,\lambda_i)$, $i\!=\!1,2$, with a probability density function (pdf) and a complementary cumulative distribution function (ccdf) given, respectively, by \cite{Digital_comm_fading_Alouini2005}\vspace{-0.15cm}
\begin{equation}
f_{v_i}(x)=\frac{\lambda_i^{m_iN_i}}{\Gamma(m_iN_i)}x^{m_iN_i-1}\e^{-\lambda_i x}
\label{eq:pdf}
\end{equation}\vspace{-0.2cm}
and \vspace{-0.2cm}
\begin{equation}
\bar{F}_{v_i}(x)=\e^{-\lambda_i x}\sum\limits_{k=0}^{m_iN_i-1}\frac{(\lambda_i x)^k}{k!}.
\label{eq:ccdf}\vspace{-0.1cm}
\end{equation}
In the equivalent single antenna system, the PS transmits an energy signal of power $\Pt$ in the WPT subslot for a duration of $\tau T$ through a DL channel with power gain $v_1$. At the WD, the received signal of power $\PR\!=\!\Pt v_1$ is applied to a non-linear EH circuit and the  harvested power is modelled by (\ref{eq:EH_non_linear}). The harvested energy $\EEH\!=\!\PEH\tau T$ is fully consumed in the UL subslot for a duration of $(1\!-\!\tau)T$. In particular, assuming that the WD uses a power amplifier with efficiency $\theta\!<\!1$, then $\theta\!\times\!100\%$ of the harvested energy is used for information transmission and the remaining amount of harvested energy is consumed by the power amplifier. Hence, the WD transmits information with power $\PWD\!=\!\theta\PEH\tau/(1\!-\!\tau)$ to the IRS through an UL channel with power gain $v_2$. At the IRS, the received signal is impaired by additive white Gaussian noise of power $\sigma^2$ and the received instantaneous SNR is $\gamma\!=\!\frac{\theta\PEH\tau v_2}{(1-\tau)\sigma^2}$.

%

\section{Outage Probability and Throughput Analysis}
\label{s:outage_analysis}
In this section, we analyze the outage probability  and the average throughput for the information transmission in the UL channel. Since the WD does not have CSI knowledge, it transmits data at a constant rate of $R$ bits/s/Hz. Therefore, assuming the use of a capacity-achieving code, an outage occurs when $R>\log_2(1+\gamma)\Rightarrow \gamma<\gammathr$, where $\gamma$ is the UL instantaneous SNR and $\gamma_{\rm thr}=2^{R}-1$ is the threshold SNR. 
\begin{proposition} \normalfont
For the considered WPC system, the outage probability at the IRS can be written as\vspace{-0.1cm}
\begin{equation}
\pout=1\!-\!\int_{0}^{\infty}\bar{F}_{v_2}\left(\frac{c(1+\e^{ab})}{1-\e^{-a\Pt x}}-c\e^{ab}\right)f_{v_1}(x)\dd x
\label{eq:Outage_general_INF_limits}
\end{equation}
or equivalently using finite integral limits as \vspace{-0.1cm}
\begin{equation}
\pout\!\!=\!1\!-\!\!\frac{1}{a\Pt}\!\!\int_{0}^{1}\!\!\!\!\bar{F}_{v_2}\!\!\left(\!\!\frac{c(1\!+\!\e^{ab}\!)}{y}\!-\!c\e^{ab}\!\!\right)\!f_{v_1}\!\!\left(\!\!\frac{-\ln(1\!-\!y)}{a\Pt}\!\!\right)\!\!\frac{1}{1\!-\!y}\dd y,\hspace{-0.2cm}
\label{eq:Outage_general_finite_limits}
\end{equation}
where $c=\frac{\gammathr\sigma^2(1-\tau)}{\theta M \tau}$  and $\gammathr=2^R\!-\!1$. Furthermore, the average throughput in bits/s/Hz is given by\vspace{-0.1cm}
\begin{equation}
T\!H=R(1-\tau)(1-\pout).
\label{eq:throughput} 
\end{equation}
\end{proposition}
\begin{proof}
For a given DL channel $v_1$, the outage probability is given by $\pout=\pr(\frac{\theta\PEH\tau v_2}{(1-\tau)\sigma^2}<\gammathr)=\pr(v_2<v_{\rm thr})=1-\bar{F}_{v_2}(v_{\rm thr})$, where $v_{\rm thr}=\frac{\gammathr \sigma^2 (1-\tau)}{\theta\PEH \,\tau}$. Using $\PEH$ in (\ref{eq:EH_non_linear}) with $\PR=\Pt v_1$ and the identity $(1-z)^{-1}=1-(1-z^{-1})^{-1}$, we get $v_{\rm thr}=\frac{c(1+\e^{ab})}{1-\e^{-a\Pt v_1}}-c\e^{ab}$, where $c=\frac{\gammathr\sigma^2(1-\tau)}{\theta M \,\tau}$. By averaging over $v_1$, the average outage probability reduces to $\pout=\int_{0}^\infty\left[ 1-\bar{F}_{v_2}\left(\frac{c(1+\e^{ab})}{1-\e^{-a\Pt x}}-c\e^{ab}\right)\right]f_{v_1}(x)\dd x$. Using $\int_0^\infty f_{v_1}(x)\dd x=1$, the outage probability reduces to (\ref{eq:Outage_general_INF_limits}). Applying the change of variables $y=1-\e^{-a\Pt x}$ in (\ref{eq:Outage_general_INF_limits}), we obtain the equivalent representation of the outage probability with finite integral limits in (\ref{eq:Outage_general_finite_limits}). This completes the proof.
\end{proof}

\begin{proposition} \normalfont
For the considered WPC system with Nakagami-$m$ fading channels, the outage probability at the IRS is given by\vspace{-0.2cm}
\begin{align}
&\pout\!\!=\!1\!-\!\!\frac{\lambda_1(-1)^{m_1N_1\!-\!1}}{a\Pt\Gamma(m_1N_1)}\e^{-\lambda_2c}\!\!\!\sum\limits_{k=0}^{m_2N_2-1}\!\!\frac{(\lambda_2c\e^{ab})^k}{k!}\!\sum\limits_{l=0}^{k}\!\!\binom{k}{l}\!(-\!1)^{k\!-\!l}\notag\\[-1ex]
&(1\!+\!\e^{-ab})^l \frac{\partial^{m_1N_1\!-\!1} }{\partial s^{m_1N_1\!-\!1}}\!\left[\Gamma\!\!\left(\frac{\lambda_1 s}{a\Pt}\right)\!U\!\!\left(\frac{\lambda_1 s}{a\Pt},l,\lambda_2c(1+\e^{ab})\right)\right]\!\Bigg|_{s=\!1},
\label{eq:Outage_Nakagami}
\end{align}\\[-1.5ex]
where $c\!=\!\frac{\gammathr\sigma^2(1-\tau)}{\theta M \tau}$ and $\gammathr\!=\!2^R\!-\!1$. Furthermore, as $\Pt\!\to\!\infty$, the harvested power $\PEH\!\to\!M$, which leads to an asymptotic outage probability and an asymptotic throughput given respectively by\vspace{-0.2cm}
\begin{equation}
\pout\Big|_{\Pt\to\infty}=1-\bar{F}_{v_2}(c)=\frac{\gamma(m_2N_2,\lambda_2 c)}{\Gamma(m_2N_2)}
\label{eq:asym_Pout}\vspace{-0.2cm}
\end{equation}
and\vspace{-0.3cm}
\begin{equation}
T\!H\Big|_{\Pt\to\infty}\!\!\!=R(1-\tau)\bar{F}_{v_2}(c)=R(1-\tau)\frac{\Gamma(m_2N_2,\lambda_2 c)}{\Gamma(m_2N_2)}.
\label{eq:asym_Throughput}
\end{equation}
\label{prop:outage_Nakagami}\vspace{-0.3cm}
\end{proposition}
\begin{proof}
The proof is provided in Appendix \ref{App:proof_outage_Nakagami}.
\end{proof}
\begin{remark}
We note that (\ref{eq:Outage_Nakagami}) is a closed-form expression for the outage probability for a single-antenna Rayleigh fading DL WPT channel, i.e., for $m_1\!=\!1$ and $N_1\!=\!1$, since $\frac{\partial^{0}x}{s^0}\!=\!x$. When $m_1N_1\!\neq\!1$, there is no simple closed-form expression for the derivative of the term in the square brackets in (\ref{eq:Outage_Nakagami}). Hence, we evaluate it through numerical differentiation. We use the numerical differentiation solver of Mathematica that is based on Richardson's extrapolation method, which is a computationally fast and efficient method that provides accurate results for relatively high-order derivatives. 
\end{remark}
\begin{remark}
We note that, unlike the linear EH model which results in zero asymptotic outage probability, see. e.g. \cite[Fig. 2]{CR:Perf_Analys_WPT_Nakagami_linear}, the non-linear EH model results in an outage floor for high PS transmit powers $\Pt$, cf. (\ref{eq:asym_Pout}). Hence, the linear EH model suggests a misleadingly optimistic outage performance which is not achievable with practical RF EH circuits. Moreover, we note that the asymptotic relations in (\ref{eq:asym_Pout}) and  (\ref{eq:asym_Throughput}) are independent of the WPT DL channel. That is, when the PS transmit power $\Pt$ and the DL channel gain $v_1$ are such that the harvested power saturates, i.e.,  $\PEH\to M$, then the WD transmits with a constant power given by $\PWD=\theta M\tau/(1-\tau)$ and  the WPC system reduces asymptotically to a WIT UL channel with constant power supply. We note that our asymptotic analysis in Proposition \ref{prop:outage_Nakagami} and in the following Corollary is valid for any non-linear EH model whose harvested power saturates at $M$.
\label{rem:comments_on_Pout}
\end{remark}
\begin{corollary}\normalfont
The optimal rate that maximizes the asymptotic throughput in (\ref{eq:asym_Throughput}) is unique and satisfies
\begin{equation}
R=\frac{\Gamma\left(m_2N_2,\alpha(2^R-1)\right)\e^{\alpha(2^R-1)}}{\alpha^{m_2N_2}\ln(2) 2^R (2^R-1)^{m_2N_2-1}},
\label{eq:opt_R}
\end{equation}
where $\alpha\!=\!\lambda_2\sigma^2(1-\tau)/(\theta M\tau)$, and the optimal EH time factor that maximizes the asymptotic throughput in (\ref{eq:asym_Throughput}) is unique and satisfies
\begin{equation}
\tau=\frac{\left(\beta\frac{1-\tau}{\tau}\right)^{m_2N_2}\e^{-\beta\frac{1-\tau}{\tau}}}{\Gamma\left(m_2N_2,\beta\frac{1-\tau}{\tau}\right)},
\label{eq:opt_tau}
\end{equation}
where $\beta=\lambda_2\sigma^2(2^R-1)/(\theta M)$.
\label{coro:optimal_tau_optimal_R}
\end{corollary}
\begin{proof}
Eqs. (\ref{eq:opt_R}) and (\ref{eq:opt_tau}) can be obtained by taking the first-order derivative of the asymptotic throughput in (\ref{eq:asym_Throughput}) with respect to $R$ and $\tau$, respectively, and setting it to zero. It can be shown that the throughput function in (\ref{eq:asym_Throughput}) is strictly quasi-concave in $R$ and in $\tau$. Hence, the solutions of (\ref{eq:opt_R}) and (\ref{eq:opt_tau}) are the unique globally optimal rate and EH time factor, respectively, that maximize the asymptotic throughput in (\ref{eq:asym_Throughput}).
\end{proof}
%

\section{Simulation and Numerical Results}
\label{s:simulation_results}
In this section, we evaluate the performance of the considered WPC system and validate our analysis through simulations.  The simulation parameters are summarized in Table \ref{tab:simulation_parameters}.  To validate the accuracy of the adopted non-linear EH model, we simulate the system with the interpolated transfer function of the measured data from the EH circuit in \cite{self_calibrating_EH_circuit2013} and the fitted non-linear EH model in (\ref{eq:EH_non_linear}),  cf. Fig. \ref{fig:transfer_functions}. Figs. \ref{fig:Outage_Nakagami_diff_antennas_sim_anal_measured} -- \ref{fig:Max_Throughput_vs_tau_diff_N_diff_m} show that the analytical results derived in Section \ref{s:outage_analysis} are in perfect agreement with the simulated results that use the non-linear EH model.
\begin{table}[!tp]
\vspace{0.18in} 
\caption{Simulation Parameters}
\begin{tabular}{@{}ll@{}}  \toprule   
Parameter & Value \\ \midrule \addlinespace[0.5em]
Distance between PS and WD & $d_1=4\,$m\\ \addlinespace[0.5em]
Distance between WD and IRS & $d_2=10\,$m\\ \addlinespace[0.5em]
UL and DL carrier frequency  & $868\,$MHz, same frequency as in \cite{self_calibrating_EH_circuit2013} \\ \addlinespace[0.5em]
Noise power at the IRS & $\sigma^2=-96\,$dBm \\ \addlinespace[0.5em]
Power amplifier efficiency at WD & $\theta=0.5$ \\ \addlinespace[0.5em]
DL and UL path loss exponent & $2.8$ \\ \addlinespace[0.5em]
Antenna gains at PS, IRS, WD & $11\,$dBi, $11\,$dBi, and $3\,$dBi\\ \addlinespace[0.5em]
Non-linear EH model fitted to \cite{self_calibrating_EH_circuit2013} & $a\!=\!47083$, $b\!=\!\unit[2.9]{\mu W}$, $M\!=\!\unit[9.079]{\mu W}$\\
 \addlinespace[0.5em]
\bottomrule
\end{tabular}
\label{tab:simulation_parameters}
\end{table}
\begin{figure}[!t] 
\centering
\includegraphics[width=0.45\textwidth, height=0.185\textheight]{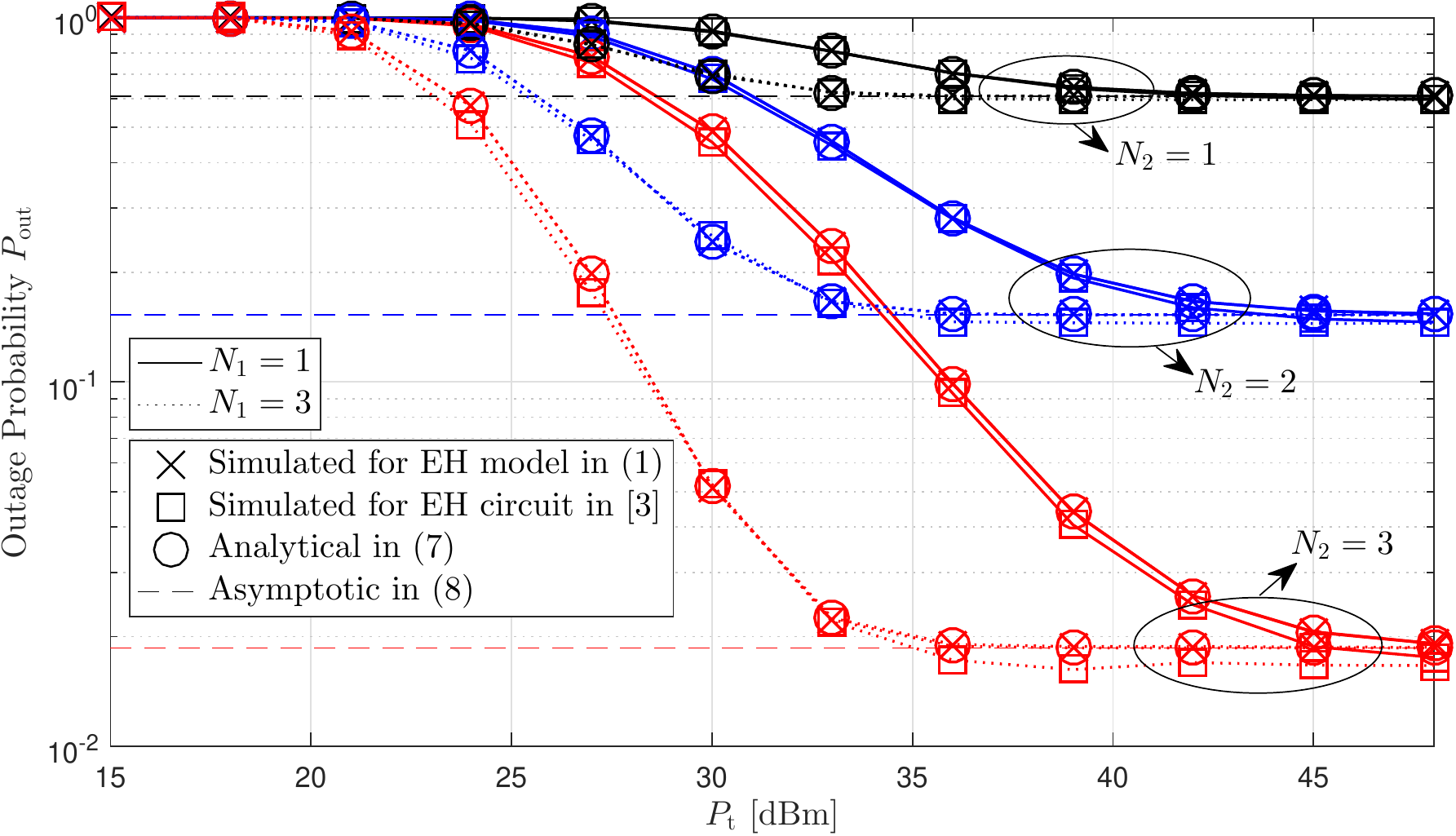}\\[-1ex]
\caption{Outage probability vs. PS transmit power for $m_1\!=\!m_2\!=\!2$, $\tau\!=\!0.5$, $R\!=\!\unit[5]{bits/s/Hz}$, and different $N_1$ and $N_2$.}
\label{fig:Outage_Nakagami_diff_antennas_sim_anal_measured}\vspace{-0.5em}
\end{figure}
\begin{figure}[!t] 
\centering
\includegraphics[width=0.45\textwidth, height=0.185\textheight]{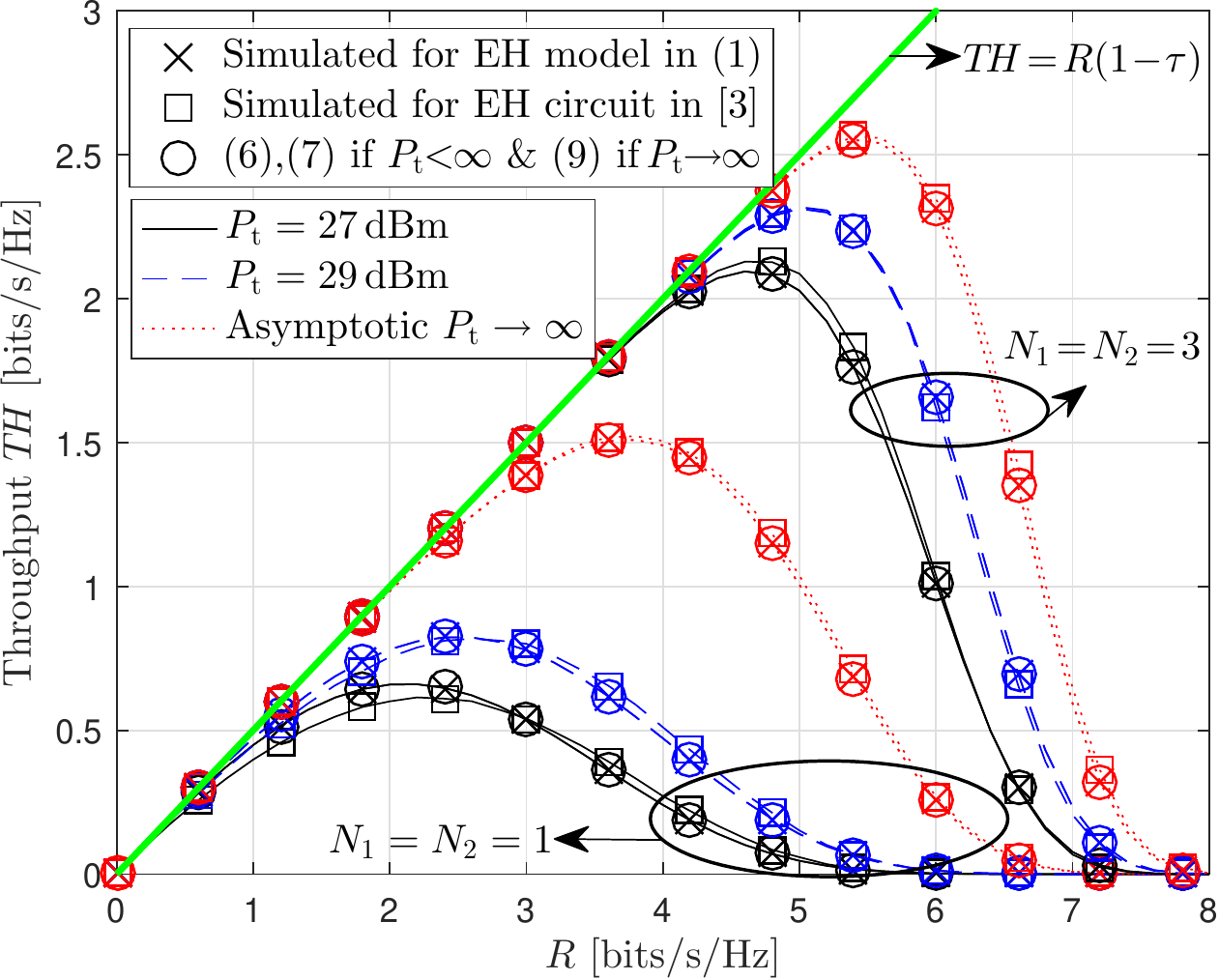}\\[-1ex]
\caption{Throughput vs. rate $R$  for $m_1=m_2=2$, $\tau=0.5$, different transmit powers $\Pt$, and different $N_1$ and $N_2$.} 
\label{fig:Throughput_vs_rate_diff_powers}
\end{figure}

In Fig. \ref{fig:Outage_Nakagami_diff_antennas_sim_anal_measured}, we plot the outage probability vs. the PS transmit power $\Pt$ for $m_1\!=\!m_2\!=\!2$, $\tau\!=\!0.5$, $R\!=\!\unit[5]{bits/s/Hz}$, and different values of $N_1$ and $N_2$. Unlike with the linear EH model \cite[Fig. 2]{CR:Perf_Analys_WPT_Nakagami_linear}, the non-linear EH model exhibits an outage floor due to the saturation of the harvested power for high PS transmit powers. Moreover, we observe that increasing $N_1$ and $N_2$ have different effects on the outage probability. In particular, for a larger number of energy beamforming antennas $N_1$, a lower PS transmit power is required to reach the outage floor. However, the asymptotic outage performance cannot be improved by increasing the PS transmit power nor by increasing the number of energy beamforming antennas $N_1$, cf. Remark \ref{rem:comments_on_Pout}. On the other hand, increasing the number of MRC receive antennas $N_2$ improves the outage floor significantly. For example, the asymptotic outage rate improves from $\!60\,\%\!$ to only $\!2\,\%\!$ by increasing $\!N_2\!$ from $\!N_2\!=\!1$ to $\!N_2\!=\!3$. Furthermore, the outage performance for the interpolated data of the EH circuit in \cite{self_calibrating_EH_circuit2013} is shown to be very close to that of the non-linear EH model in (\ref{eq:EH_non_linear}), but it exhibits a slightly lower outage floor. This is because the saturation level of the EH circuit is slightly higher than that of the non-linear EH model, cf. Fig. \ref{fig:transfer_functions}.

In Fig. \ref{fig:Throughput_vs_rate_diff_powers}, we plot the throughput vs. the transmission rate $R$ for $m_1\!=\!m_2\!=\!2$ and $\tau\!=\!0.5$. The throughput performance improves with increasing the PS transmit power until it reaches the asymptotic throughput shown in red. Moreover, a significant throughput gain is achieved by increasing the number of antennas. For example, when the number of antennas is increased from $N_1\!=N_2\!=\!1$ to $N_1\!=N_2\!=\!3$, the maximum throughput is increased by $228.5\%$ for $\Pt=\unit[27]{dBm}$. Besides, we also show an upper bound on the achievable throughput given by $T\!H=R(1-\tau)$, which is the asymptotic throughput for the linear EH model since its asymptotic outage probability is zero.

\begin{figure}[!t] 
\centering
\includegraphics[width=0.45\textwidth, height=0.185\textheight]{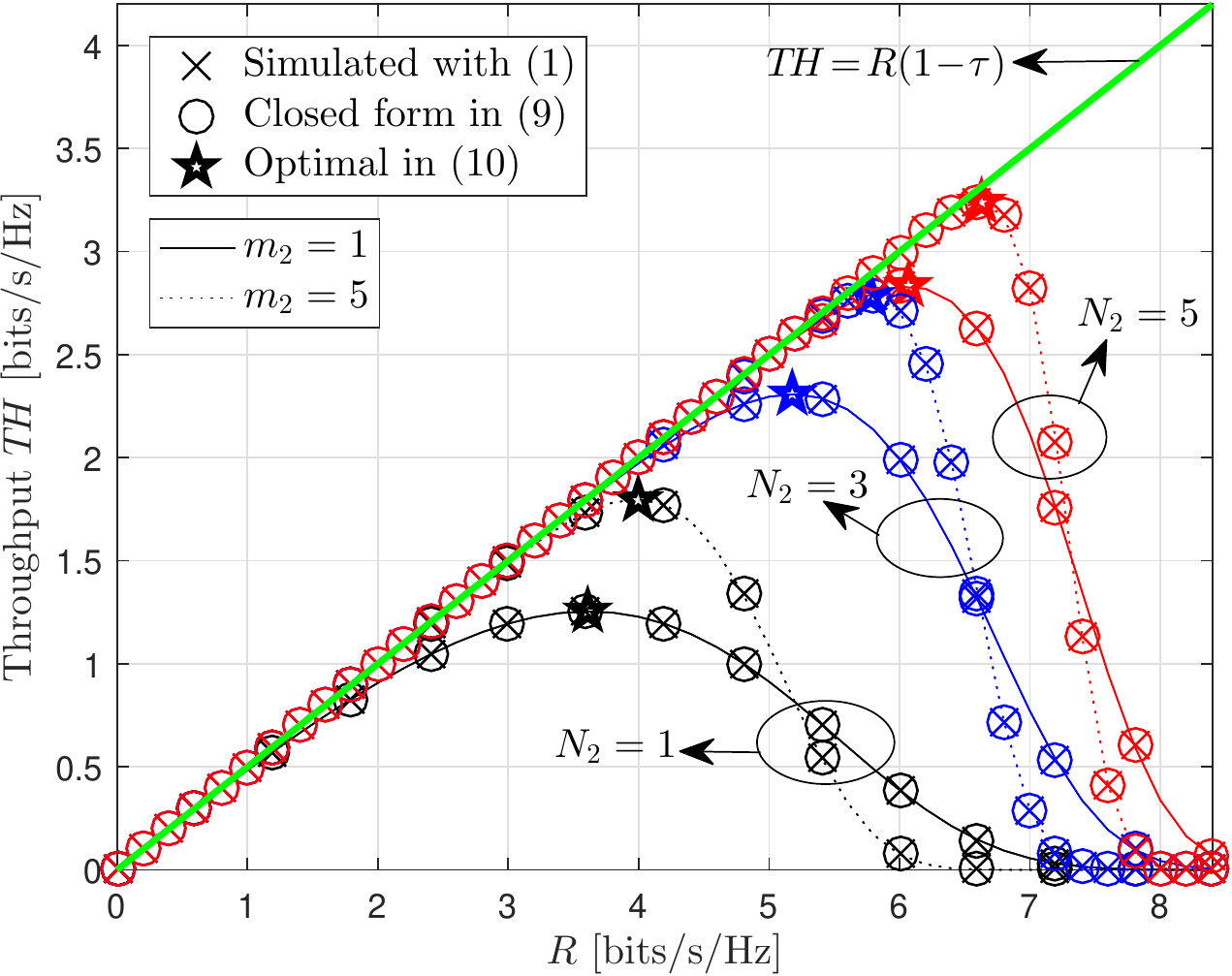}\\[-1ex]
\caption{Asymptotic throughput vs. rate for $\tau=0.5$ and different $m_2$ and $N_2$.}
\label{fig:Max_Throughput_vs_rate_diff_N_diff_m}\vspace{-0.5em}
\end{figure}
\begin{figure}[!t] 
\centering
\includegraphics[width=0.45\textwidth, height=0.1605\textheight]{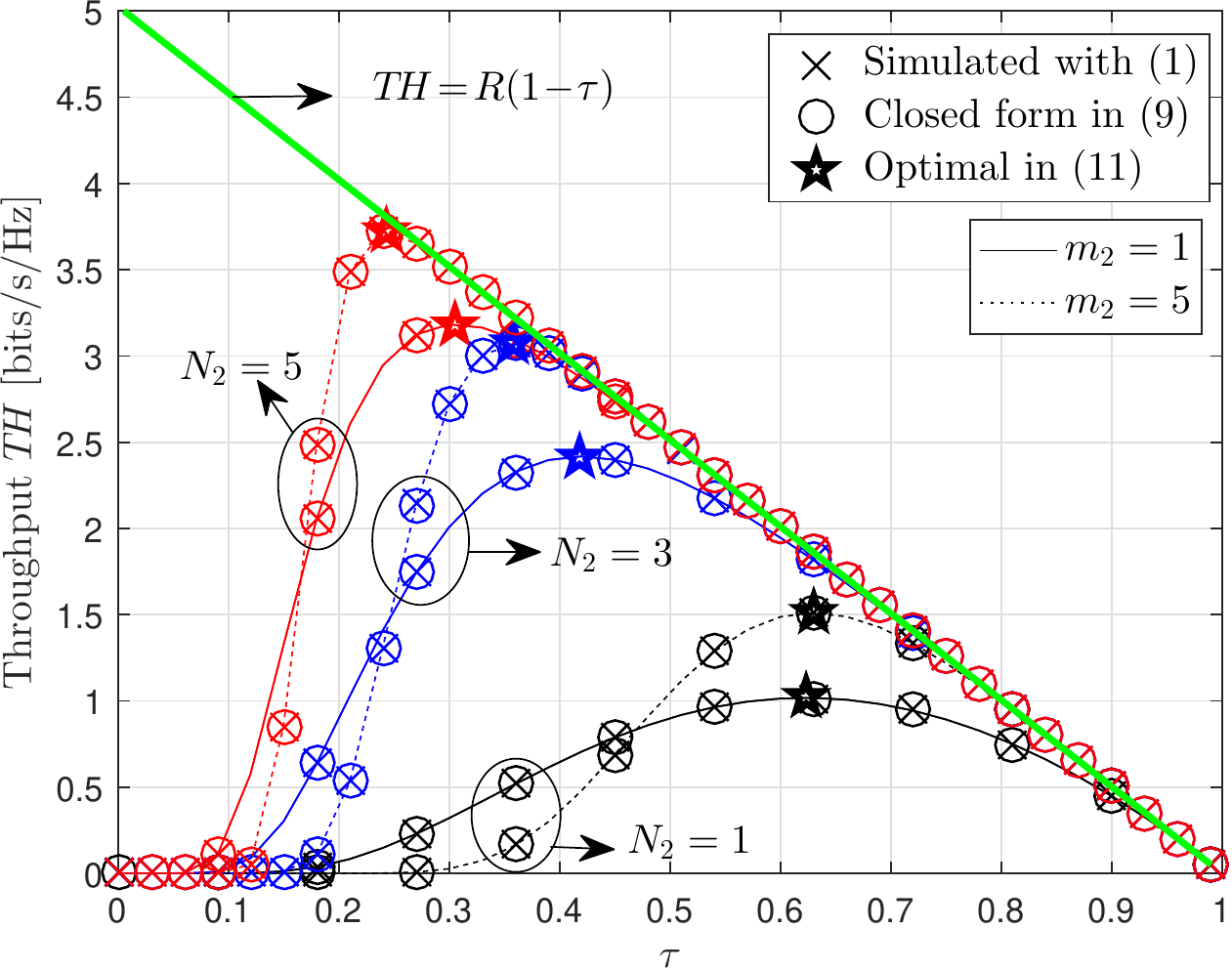}\\[-1ex]
\caption{Asymptotic throughput vs. EH time factor $\tau$  for $R=\unit[5]{bits/s/Hz}$ and different $m_2$ and $N_2$.\\[-5ex]}
\label{fig:Max_Throughput_vs_tau_diff_N_diff_m}
\end{figure}

In Figs. \ref{fig:Max_Throughput_vs_rate_diff_N_diff_m} and \ref{fig:Max_Throughput_vs_tau_diff_N_diff_m}, we show the behaviour of the \emph{asymptotic} throughput as a function of the transmission rate $R$ and the EH time factor $\tau$, respectively, for different $m_2$ and $N_2$. Recall that the asymptotic behaviour of the system is independent of the WPT DL channel, cf. Remark \ref{rem:comments_on_Pout}. We observe that the larger the number of receive antennas $N_2$, the longer the throughput curve follows the upper bound $T\!H\!=\!R(1-\tau)$. In other words, as $N_2\to\infty$, the asymptotic throughput tends to that of the linear EH model. Furthermore, although $m_2\!=\!1$ represents a more severe fading of the UL channel compared to $m_2\!=\!5$, it may result in a superior throughput performance for high rates $R$ or small EH time factors $\tau$. This is because smaller values of $m_2$ lead to more randomness of the UL channel $v_2$ which results in a better outage probability $\pout\!=\!\pr(v_2\!<\!c)$ when the threshold $c\!=\!\frac{(2^R-1)\sigma^2(1-\tau)}{M \tau}$ is large compared to the mean $\mu_2$ of $v_2$, see \cite[Fig. 2.1]{Digital_comm_fading_Alouini2005}. Finally, we validate the optimal rate $R$ and EH time factor $\tau$ that maximize the asymptotic throughput, by solving (\ref{eq:opt_R}) and (\ref{eq:opt_tau}) using a standard numerical root finding  tool, cf. \quotes{\large $\star$\normalsize} in Figs. \ref{fig:Max_Throughput_vs_rate_diff_N_diff_m} and \ref{fig:Max_Throughput_vs_tau_diff_N_diff_m}. We observe that as the number of MRC receive antennas $N_2$ increases, the optimal transmission rate $R$ increases, whereas the optimal EH time factor $\tau$ decreases.\vspace{-0.3cm}
\section{Conclusion}
\label{s:conclusion}
In this paper, we analyzed the outage probability and the average throughput of a multi-antenna WPC system with a non-linear EH model for Nakagami-$m$ fading. We have shown that for high PS transmit powers, the outage probability of the WIT link saturates and the asymptotic system performance is independent of the WPT link. Moreover, our results reveal that increasing the number of beamforming antennas, reduces the PS  transmit power required for the outage probability to saturate, whereas increasing the number of MRC receive antennas improves the asymptotic system performance significantly.


\appendices
\section{Proof of Proposition \ref{prop:outage_Nakagami}}
\label{App:proof_outage_Nakagami}
By using the pdf and the ccdf of the equivalent UL and DL channels from (\ref{eq:pdf}) and (\ref{eq:ccdf}) in (\ref{eq:Outage_general_finite_limits}) and defining $c_1=\lambda_1/(a\Pt)$ and $c_2=\lambda_2c(1+\e^{ab})$, we can write
$\bar{F}_{v_2}\!\!\left(\!\frac{c(1+\e^{ab}\!)}{y}\!-\!c\e^{ab}\!\right)=\e^{\lambda_2 c\e^{ab}}\!\e^{\frac{-c_2}{y}}\sum\limits_{k=0}^{m_2N_2-1}\!\!\frac{(\lambda_2c\e^{ab})^k}{k!}\!\!\sum\limits_{l=0}^k\!\!\binom{k}{l}(-1)^{k-l}(1\!+\!\e^{-ab})^l y^{-l}$ and
$f_{v_1}\!\!\left(\!\frac{-\ln(1\!-\!y)}{a\Pt}\!\right)=\frac{\lambda_1^{m_1N_1}}{\Gamma(m_1N_1)}\left(\frac{-\ln(1\!-\!y)}{a\Pt}\right)^{m_1N_1-1}(1-y)^{c_1}$, where we used the identities $\alpha\ln(z)=\ln(z^\alpha)$ and $\e^{\ln(z)}=z$. Thus, the outage probability in (\ref{eq:Outage_general_finite_limits}) reduces to\vspace{-0.2cm}
\begin{equation}
\begin{aligned}
&\pout\!\!=\!1\!-\!\!\frac{c_1(-1)^{m_1N_1\!-\!1}}{\Gamma(m_1N_1)}\e^{\lambda_2c\e^{ab}}\!\!\!\sum\limits_{k=0}^{m_2N_2-1}\!\!\frac{(\lambda_2c\e^{ab})^k}{k!}\!\sum\limits_{l=0}^{k}\!\!\binom{k}{l}\!(-\!1)^{k\!-\!l}\\[-1ex]
&(1\!+\!\e^{-ab})^l \!\!\int_{0}^{1}\!\!\left[c_1\ln(1-y)\right]^{\!m_1N_1\!-\!1}(1\!-\!y)^{c_1-1}\e^{\frac{-c_2}{y}}y^{-l} \dd y.\!\\[-2ex]
\end{aligned}
\label{eq:Pout_step}
\end{equation}
In order to solve the integral in (\ref{eq:Pout_step}), we use the relation $\frac{\partial^{n}}{\partial s^{n}}Z^{q s}=[q\ln(Z)]^n  Z^{q s}$ at $Z\!=\!1\!-\!y$, $n\!=\!m_1N_1\!-\!1$, and $q\!=c_1$ and use the following substitution $\left[c_1\ln(1\!-\!y)\right]^{\!m_1N_1\!-\!1}(1\!-\!y)^{c_1}=\frac{\partial^{m_1N_1\!-\!1} (1-y)^{c_1s}}{\partial s^{m_1N_1\!-\!1}}\!\Big|_{s=\!1}$. By swapping the order of integration and differentiation, the integral in (\ref{eq:Pout_step}) reduces to $\frac{\partial^{m_1N_1\!-\!1} }{\partial s^{m_1N_1\!-\!1}}I_s\big|_{s=\!1}$, where $I_s=\int_{0}^{1}(1\!-\!y)^{c_1s-1}\e^{\frac{-c_2}{y}}y^{-l} \dd y$. Using \cite[Eq. 3.471.2]{table_of_integrals_Ryzhik}, $I_s$ reduces to $I_s=c_2^{-\frac{l}{2}}\e^{-\frac{c_2}{2}}\Gamma(c_1 s)W_{\frac{l}{2}-c_1 s,\frac{l}{2}-\frac{1}{2}}(c_2)$, where we used $W_{\alpha,\beta}(z)=W_{\alpha,-\beta}(z)$ given in \cite[Eq. 9.232.1]{table_of_integrals_Ryzhik}. Using the relation between the Whittaker W function and the confluent hypergeometric function $W_{\alpha,\beta}(z)=z^{\beta+\frac{1}{2}}\e^{-\frac{z}{2}}U(\beta-\alpha+1/2,2\beta+1,z)$ \cite[Eq. 9.220.2]{table_of_integrals_Ryzhik}, $I_s$ reduces to $I_s=\e^{-c_2}\Gamma(c_1 s)U(c_1 s,l,c_2)$. Replacing the integral in (\ref{eq:Pout_step}) by $\frac{\partial^{m_1N_1\!-\!1} }{\partial s^{m_1N_1\!-\!1}}I_s\big|_{s=\!1}$ and using $c_1=\lambda_1/(a\Pt)$ and $c_2=\lambda_2c(1+\e^{ab})$, the outage probability reduces to  (\ref{eq:Outage_Nakagami}). Next, we prove the asymptotic outage probability in (\ref{eq:asym_Pout}). In the limit as $\Pt\to\infty$, the harvested power saturates at $\PEH\to M$. Hence, the outage probability simplifies to  $\pout=\pr(\frac{\theta M\tau v_2}{(1-\tau)\sigma^2}<\gammathr)=\pr(v_2<c)=1-\bar{F}_{v_2}(c)$, where $\bar{F}_{v_2}(c)$ is given in (\ref{eq:ccdf}). Finally, the asymptotic throughput follows directly from (\ref{eq:throughput}). This completes the proof.
\bibliographystyle{IEEEtran}
\bibliography{literature}

\end{document}